\documentclass{article}

\usepackage{arxiv}

\usepackage[utf8]{inputenc} 
\usepackage[T1]{fontenc}    
\usepackage{hyperref}       
\hypersetup{hidelinks = true}
\usepackage{url}            
\usepackage{booktabs}       
\usepackage{amsfonts}       
\usepackage{amsmath}
\usepackage{amssymb}
\usepackage{nicefrac}       
\usepackage{microtype}      

\usepackage[]{algorithm2e}

\usepackage{fancyvrb}
\usepackage{mathtools}
\DeclarePairedDelimiter{\ceil}{\lceil}{\rceil}

\usepackage{graphicx}
\usepackage{subcaption}

\usepackage{csquotes}

\usepackage{afterpage}

\usepackage{amsthm}
\theoremstyle{plain}
\newtheorem{definition}{Definition}[section]
\newtheorem{lemma}[definition]{Lemma}
\newtheorem{proposition}[definition]{Proposition}

\newtheorem{example}[definition]{Example}
\theoremstyle{remark}
\newtheorem{remark}[definition]{Remark}

\newcommand{\IE}{\mathbb{E}}
\newcommand{\IN}{\mathbb{N}}
\newcommand{\IR}{\mathbb{R}}
\newcommand{\IZ}{\mathbb{Z}}

\newcommand{\Union}{\bigcup}

\newcommand{\qtext}[1]{\quad\text{#1}\quad} 
\newcommand{\ra}{\rightarrow}

\newcommand{\floor}[1]{\lfloor#1\rfloor}

\newcommand{\float}{\mathrm{float}}

\title{circllhist}

\subtitle{A Log-Linear Histogram Data Structure for IT Infrastructure Monitoring}

\author{
  Heinrich Hartmann \\
  \texttt{heinrich.hartmann@circonus.com} \\
  Circonus \\
  \And
  Theo Schlossnagle \\
  \texttt{theo.schlossnagle@circonus.com} \\
  Circonus
}

\begin{document}

\maketitle

\begin{abstract}
  The circllhist histogram is a fast and memory efficient data structure for summarizing large
  numbers of latency measurements.  It is particularly suited for applications in IT infrastructure
  monitoring, and provides nano-second data insertion, full mergeability, accurate approximation of
  quantiles with a-priori bounds on the relative error.

  Open-source implementations are available for C/lua/python/Go/Java/JavaScript.
\end{abstract}

\section{Introduction}

Latency measurements have become an important part of IT infrastructure and application monitoring.
The latencies of a wide variety of events like requests, function calls, garbage collection, disk
IO, system-call, CPU scheduling, etc. are of great interest of engineers operating and developing IT systems.

There are a number of technical challenges associated with managing and analyzing latency data.  The
volume emitted by a single data source can easily become very large.  Furthermore, data has to be
collected and aggregated from a large number of different sources.  The data has to be stored over
long time periods (months, years), in order to allow historic comparisons and long-term service
quality estimations (SLOs).

In order to address these challenges a compression scheme has to be applied, that drastically
reduces the size of the data to be stored and transmitted.  Such a compression scheme needs to allow
at minimum (1) arbitrary aggregation of already compressed data, (2) accurate quantile
approximations, with a-priori bounds on the relative error (3) accurate counting of requests larger
or lower than a given threshold. Furthermore it's beneficial if (4) information about the full
distribution is retained, so that general probabilistic modeling techniques can be applied.

Traditionally monitoring tools, either store raw data on which calculations are performed (e.g.
ELK\footnote{\url{https://www.elastic.co/what-is/elk-stack}}) or compute latency quantiles on each host
separately and store them as numeric time series
(e.g. statsd\footnote{\url{https://github.com/statsd/statsd}}).  Both approaches have obvious
drawbacks.  The high volume of data makes raw data storage uneconomical for sources like request
latencies, and impractical for high volume sources like function-call or system-call latencies.
Direct calculation of quantiles does not allow further aggregation, so that accurate quantiles for
the total population can not be calculated.

\begin{figure}
  \includegraphics[width=\textwidth]{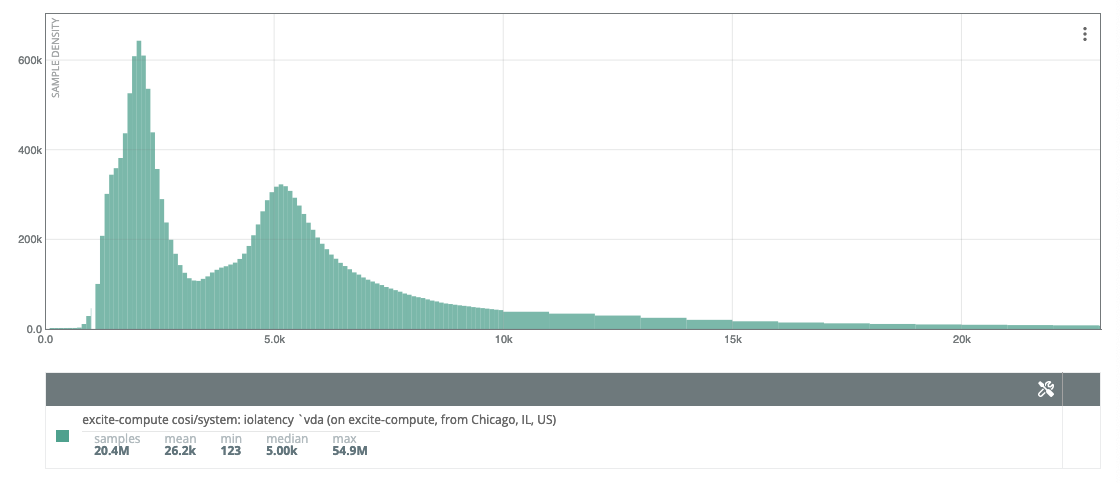}
  \caption{Circllhist representation of 20M block-level disk I/O latencies measured over the course
    of one month.}
  \label{fig:demo}
\end{figure}

Circonus has addressed this problem with the circllhist data-structure, that we describe in this document.
It has been in and production use since 2011.
Figure \ref{fig:demo} shows an example of a visualized circllhist in the Circonus product.
We have talked about this at various conferences and blogs (e.g. \cite{TS0},\cite{TS1},\cite{TS2}, \cite{HH1}).
An open source implementation is available at \cite{libcircllhist}.
However, no academic paper was published until now.

This document describes the circllhist data-structure and compares it to other methods
that have been adopted by other monitoring vendors since then.

\section{Related Work} \label{sec:rw}

There is also a fair bit of work in the academic literature on the problem of efficiently
calculating aggregated quantiles since 1980. This runs under the name ``mergeable summaries'' and
``quantile sketches''. A good summary of these methods can be found in \cite{dd}, Section 1.2.
Here we focus on methods that have been adopted in practice.

A very similar approach to ours has been suggested by G. Tene from Aszul Systems who developed
a High Definition Range (HDR) Histogram data-structure \cite{hdr} to capture latency
data for benchmarking applications.

T. Dunning and O. Ertl developed the t-digest data structure in \cite{tdigest}, which is used in the
Wavefront monitoring product. The t-digest aggregates
nearby points into clusters of adaptive size, in such a way, that high resolution data is available
at the tails of the distribution, where it's most critical for applications.

The Prometheus monitoring system \cite{prom} has added a simple histogram data-type that allows a
rough summarization of the distribution with a set of numeric time series.

Most recently Data Dog has published a logarithmic histogram data structure DDSketch in \cite{dd}.

In the next section we will develop some theory around general histogram summaries, that allows us
to precisely define HDR Histograms, DDSketches and the circllhist in section \ref{sec:eval}.

\section{Theory}

\begin{figure}
  \includegraphics[width=\textwidth]{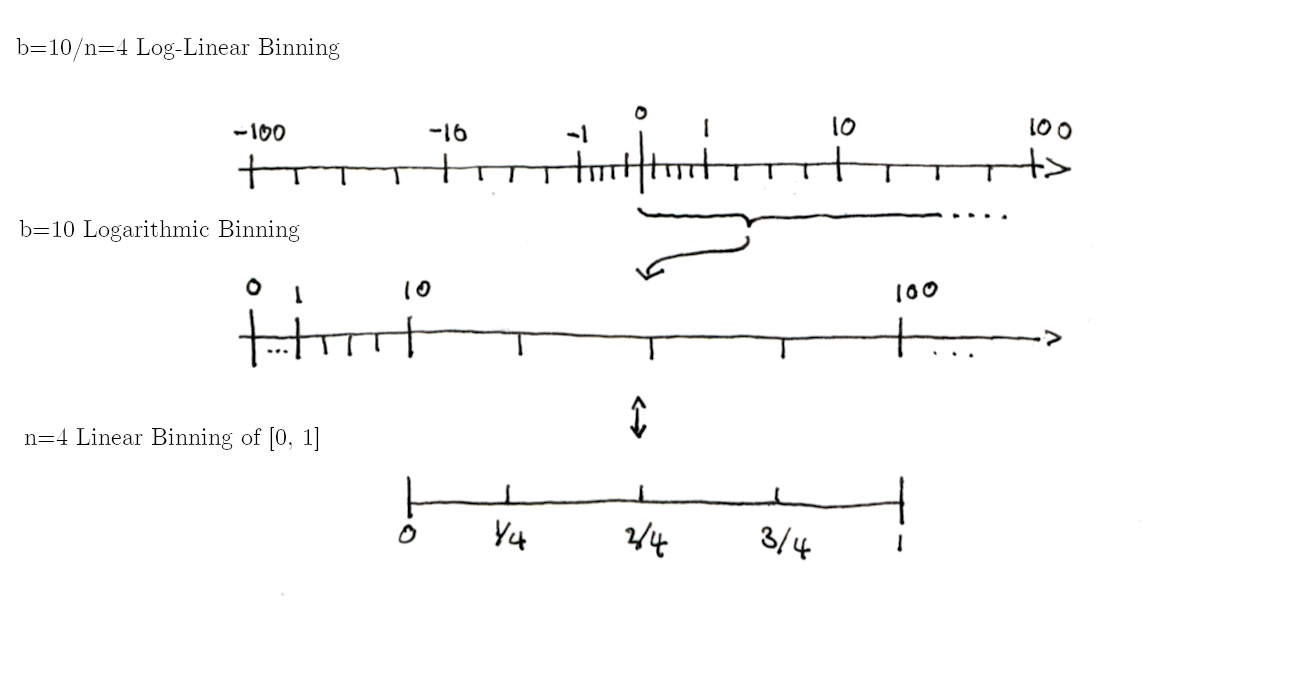}
  \caption{Construction of the Log-Linear Binning}
  \label{fig:llbins}
\end{figure}

In this section we develop an abstract theory of histograms to a degree that allows us to
formally define circllhist as a linear refinement of a logarithmic histogram structure.

The basic idea behind log-linear histograms like the circllhist is illustrated in Figure \ref{fig:llbins}.
We start with a logarithmic binning of the real axes, that has bins at the powers of ten.
\begin{center}
  \begin{BVerbatim}
    ... 0.01, 0.1, 1, 10, 100, ...
  \end{BVerbatim}
\end{center}
We divide each logarithmic bin into $n=90$ equally spaced segments. In this way the bin boundaries
are precisely the base-10, precision-2 floating point numbers:
\begin{center}
\begin{BVerbatim}
... 1.0,  1.1,  1.2,  ...   9.9,
     10,   11,   12,  ...    99,
    100,  110,  120,  ...   990, ...
\end{BVerbatim}
\end{center}
Those are the bin boundaries for the circllhist data structure.
When samples are inserted into the circllhist, we retain counts of the number of samples in each bin.
This information allows us to approximate the original location of the inserted samples with a
maximal relative error less than $5\%$.

\subsection{Binnings}

\begin{definition}
  Let $D \subset \IR$ be a connected subset of the real axes (e.g. $D=\IR, D=[0,1)$).
  A binning of $D$ is a collection of intervals $Bin[i], i \in I$, that are disjoint and collectively cover the binning domain $D$:
  \begin{align*}
    D = \Union_{i\in I} Bin[i] \qtext{and} Bin[i] \cap Bin[j] = \emptyset \qtext{for} i \neq j.
  \end{align*}
  The map that associates to each $x \in D$ the unique index $i$ so that $x \in Bin[i]$ is called
  binning map and is denoted as $bin(x) = i$.
\end{definition}

\begin{remark}
  The binning map $bin: D \ra I$ determines the binning via $Bin[i] = \{ x \in D \,|\, bin(x) = i \}$.
\end{remark}

\begin{example}
  The linear binning of $\IR$ is given by $I = \IZ$, with
  \begin{align*}
    Bin[i] = [i, i+1)  \qtext{and} bin(x)=\floor{x}
  \end{align*}
\end{example}

\begin{example}
  The length $n$ linear binning of $[0,1)$ is given by $I = \{0, \dots, n-1\}$, with
    \begin{align*}
      Bin^{Lin}_n[i]   = [ \frac{i}{n}, \frac{i+1}{n} )
      \qtext{and}
      bin^{Lin}_n(x) = \floor{x \cdot n}
    \end{align*}
\end{example}

\begin{example}\label{example:log}
  The logarithmic binning with basis $b > 0$ of $\IR_{>0}$ is given by $I=\IZ$, with
  \begin{align*}
    Bin^{Log}_b[i] = [b^i, b^{i+1})
    \qtext{and}
    bin^{Log}_b(x)=\floor{\log_b(x)}
  \end{align*}
\end{example}

\begin{definition}\label{ref}
  Given a map $\alpha: I \ra J$, and a binning $(I, Bin)$, we can define a new binning
  $(J, Bin^*)$ by setting:
  \begin{align*}
    Bin^*[j] = \Union_{i, \alpha(i) = j} Bin[i], \qtext{and} bin^*(x) = \alpha(bin(x))
  \end{align*}
  In this situation we call $(J, Bin^*)$ a coarsening of $(I, Bin)$, and $(I, Bin)$ a refinement of $(J, Bin^*)$.
\end{definition}

\begin{definition}
  Given a binning $Bin[i], i\in I$ with half-open bins $Bin[i] = [a_i, b_i)$.
  The length-n linear refinement of $(I, Bin)$, is given by the index set $I \times \{0,\dots, n-1\}$,
  bins
  \begin{align}\label{eq:lref}
    L_nBin[i,j] = [ a_i + \frac{j}{n}(b_i - a_i), a_i + \frac{j+1}{n}(b_i - a_i) )
  \end{align}
\end{definition}

\begin{lemma}
  The binning map of the length-n linear refinement is given by:
  \begin{align}\label{eq:lrefmap}
    L_nbin(x) = ( bin(x), \floor{\frac{x - a_{bin(x)}}{b_{bin(x)} - a_{bin(x)}} \cdot n } ).
  \end{align}
\end{lemma}

\begin{proof}
  We have to show that $x \in L_nBin[ L_nbin(x) ]$ for all $x \in D$.
  Let $(i,j) = L_nbin(x)$.
  Since $i = bin(x)$, we know that $x \in [a_i, b_i)$.
  Now we consider the linear map $\phi(x) = (x-a_i)/(b_i-a_i)$ which maps $Bin[i]$ bijectively to $[0,1)$.
  We have $j = \floor{ n \phi(x) }$ by definition of $L_nbin(x)$.
  To show that $x \in L_nBin[ L_nbin(x) ]$ it suffices to verify
  that $\phi(x) \in \phi(L_nBin[i,j]) = [ \frac{j}{n}, \frac{j+1}{n} )$.
  And indeed,
  \begin{align*}
    \frac{j}{n} = \frac{\floor{ n \phi(x) }}{n} \leq \phi(x) <
    \frac{\floor{ n \phi(x) } + 1 }{n} =  \frac{j+1}{n}.
  \end{align*}
\end{proof}

\begin{lemma}
  The length-n linear refinement of a binning $(I, Bin)$ is a refinement in the sense of Definition \ref{ref}.
  The index map is given by $\alpha(i,j) = i$ for $i \in I$, $j \in \{0,\dots,n-1\}$.
\end{lemma}

\begin{proof}
  We have to show that $\Union_{j} L_nBin[i,j] = Bin[i]$, for all $i \in I$.
  Again we consider the linear bijection $\phi(x) = (x-a_i)/(b_i-a_i)$, which maps $Bin[i]$ to $[0,1)$ and $L_nBin[i,j]$ to $[j/n, (j+1)/n)$.
  Hence it suffices to show that $[j/n, (j+1)/n)$ cover $[0,1)$ for $j=0,\dots,n$ which is evident.
\end{proof}

Now we are in a position to define log-linear binnings.

\begin{definition}\label{def:ll}
  The base $b$, length $n$ log-liner binning of $\IR_{>0}$ is the length-n linear refinement of the base-b Logarithmic binning of $\IR_{>0}$.
\end{definition}

\begin{proposition}\label{prop:ll}
  Let $b,p$ be positive integers. The boundaries of the base $b$, length $n = b^p - b^{p-1}$ log-linear
  binning or $\IR_{>0}$ are precisely the base-p precision-p floating point numbers:
  \begin{align*}
    \float_{b,p}(e, d) = \frac{d}{b^{p-1}} \cdot b^e = d \cdot b^{e-p+1} \qtext{with} e \in \IZ, d \in \{ b^{p-1}, \dots, b^p - 1 \}
  \end{align*}
  The binning map is given by
  \begin{align*}
    bin(x) = (e(x), d(x) - b^{p-1}), \qtext{with} e(x) = \floor{\log_b(x)},\; d(x) = \floor{x \cdot b^{-e(x) + p - 1}}
  \end{align*}
  with values $e(x) \in \IZ$ and $d(x) \in \{b^{p-1}, \dots, b^{o} - 1\}$.
  This binning is also called base-b precision-p log-linear binning.
\end{proposition}

\begin{example}
  According to Proposition \ref{prop:ll}, the base-10 precision-1 binning has bin boundaries at
  \begin{align*}
      \{ d \cdot 10^e  \,|\, e \in \IZ, d \in \{ 1, \dots, 9 \} \} = \{ \dots 0.8, 0.9,\; 1, 2 \dots, 8, 9,\; 10, 20 \dots \}
  \end{align*}
  binning map
  \begin{align*}
    bin(x) = (e(x), d(x) - 1), \qtext{with} e(x) = \floor{\log_{10}(x)},\; d(x) = \floor{x / 10^{e(x)}}.
  \end{align*}
\end{example}

\begin{proof}
  To proof Proposition \ref{prop:ll}, we compute the log-linear bin boundaries using equation \ref{eq:lref}
  with $a_i = b^i, b_i = b^{i+1}$ and $n = b^p - b^{p-1}$:
  \begin{align*}
    Bin[e,j] &= [ b^e + \frac{j}{b^p - b^{p-1}}(b^{e+1} - b^e), b^e + \frac{j + 1}{b^p - b^{p-1}}(b^{e+1} - b^e) ) \\
      &= [ b^{e-p+1}(b^{p-1} + j), b^{e-p+1}(b^{p-1} + j + 1) ) \\
      &= [ d b^{e-p+1}, (d + 1) b^{e-p+1} )
  \end{align*}
  where we set $d = b^{p-1} + j$. If $j$ runs through $1 \dots n$, then $d$ runs through $b^{p-1},\dots, b^p -1$.
  This shows that the lower boundaries are exactly the base-b precision-p floating point numbers.
  For the upper boundary note, that the if $d=b^p-1$ then $(d+1)b^{e-p+1} = b^{p-1} b^{(e+1)-p+1}$ is again
  a base-b precision-p floating point number (with a larger exponent).

  The binning map can be explicitly calculated using Equation \ref{eq:lrefmap} as:
  \begin{align*}
    bin(x) &= (e(x), k(x)), \qtext{with} e(x) = \floor{\log_b(x)} \\
    k(x) & = \floor{ \frac{x - b^{e}}{b^{e+1} - b^e} (b^p - b^{p-1}) }
    = \floor{x \cdot b^{-e + p - 1}} - b^{p-1}
    = d(x) - b^{p-1}
  \end{align*}
  as claimed.
\end{proof}

\begin{definition}\label{def:circllhist}
  The circllhist binning is the base-10 precision-2 log-linear binning extended to the real axes,
  with bins:
  \begin{align*}
    Bin[+1,e,d] &= [ d \cdot {10}^{e-1}, (d + 1) 10^{e-1} ), \quad e \in \IZ, d \in \{ 10, \dots, 99 \} \\
    Bin[0,0,0]  &= \{ 0 \} \\
    Bin[-1,e,d] &= [ -d \cdot {10}^{e-1}, -(d + 1) \cdot 10^{e-1} ).
  \end{align*}
  The binning map is given by:
  \begin{align*}
    bin(x)  &= (+1, e, d), e = \floor{\log_{10}(x)}, d = \floor{x \cdot 10^{-e - 1}} \\
    bin(0)  &= (0, 0, 0) \\
    bin(-x) &= (-1, e, d)
  \end{align*}
  for $x > 0$.
\end{definition}

\begin{proposition} \label{prop:rec}
  The binning map of the circllhist can be recursively computed with the following algorithm:

\begin{BVerbatim}[fontfamily=tt]
  function bin(x)
    if x == 0:
      return (0,0,0)
    if x < 0:
      (s, e, d) := bin(-x)
      return (-1, e, d)
    if x < 10:
      return bin(x * 10) - (0, 1, 0)
    if x > 100:
      return bin(x / 10) + (0, 1, 0)
    else: # 10 <= x < 100:
      return (+1, 1, floor(x))
  end
\end{BVerbatim}

In particular, the circllhist binning map can be computed without use of the logarithm function.
\end{proposition}

\begin{proof}
  The recursion terminates since every positive number $x$ can be brought into range $10 \leq x <
  100$ with a finite number of divisions or multiplications by $10$.  It's straight forward to
  verify that each case computes valid results assuming that the results of the recursive call are
  correct.
\end{proof}

\begin{remark}
  If $x,e$ are integers, then the binning map of the number $x \cdot 10^{e}$ can be computed without
  the use of floating point arithmetic as $bin(x) + (0,e,0)$ using Algorithm \ref{prop:rec}.

  This is of practical relevance when used in an environment which do not have floating
  point arithmetic available. One example being nano-second latencies measured in the Linux kernel,
  or embedded devices.
\end{remark}

\subsection{Paretro Midpoints}

\begin{proposition} \label{prop:pdist}
  Given an interval $[a,b]$ in $\IR_{>0}$ the unique point $m$ in $[a,b]$ so that the maximal
  relative distance $rd(m, y) = |m-y|/y$ to all other points in $[a,b]$ is minimized
  by the paretro midpoint
  $m = 2ab / (a + b)$.

  The maximal relative distance to the paretro midpoint is assumed at the interval boundaries
  $rd(m,a) = rd(m,b) = (b - a) / (a + b)$.
\end{proposition}

\begin{proof}
  We have to minimize the function $maxrd(x) = max_{y\in[a,b]} rd(x, y)$ over $[a,b]$.
  The maximum $max_{y\in[a,b]} rd(x, y)$ is attained either for $y = a$ or $y = b$,
  hence $maxrd(x) = max\{ rd(x,a), rd(x,b) \}$.

  Note that the function $f(x) = rd(x, a)$ is continues and strictly monotonically increasing on $[a,b]$, with $f(a) = 0, f(b) > 0$,
  and the function $g(x) = rd(x, b)$ is continues and strictly monotonically decreasing on $[a,b]$, with $g(a) > 0$ and $g(b) = 0$.
  The point $m$ is the unique point in $[a,b]$ where both functions are equal, with
  \begin{align*}
    rd(m, a) = \frac{b - a}{a + b} = rd(m, b)
  \end{align*}
  Now if $a \leq x < m$ then $rd(x, b) = g(x) > g(m)$ and so $maxrd(x) > maxrd(m) = g(m)$.
  Similarly if $m < x \leq b$ then $rd(x, a) = f(x) > f(m)$ and so $maxrd(x) > maxrd(m) = f(m)$.

  This shows that $x=m$ is the unique minimum of $maxrd(x)$ on $[a,b]$.
\end{proof}

The following proposition gives a probabilistic interpretation of the location of relative distance minimizing midpoint.
Recall, that the expected value of a uniformly distributed random variable $X \sim U[a,b]$ is the midpoint $\IE[X] = (a+b)/2$.

\begin{proposition}
  Given an interval $[a,b]$, and an a=2 paretro distributed random variable $X$, then
  the paretro midpoint is the conditional expectation:
  \begin{align*}
    \IE[ X \, | \, X \in [a,b] \,] = 2ab / (a + b).
  \end{align*}
\end{proposition}

\begin{proof}
  The paretro distribution has density $p(x) =C \cdot 1/x^{a+1}$, for some positive constant $C$, so for $a=2$ we get
  \begin{align*}
    \IE[ X \, | \, X \in [a,b] \,] = \frac{\int_a^b x p(x) dx}{\int_a^b p(x) dx}
    = \frac{\int_a^b x^{-2}  dx}{\int_a^b x^{-3} dx} = 2 \frac{b^{-1} - a^{-1}}{b^{-2}- a^{-2}}
    = \frac{2}{b^{-1} + a^{-1}}
    = 2 \frac{ab}{a + b}
  \end{align*}
  which proves the claim.
\end{proof}

\begin{proposition}\label{prop:21}
  The maximal relative distance to a paretro midpoint in the circllhist binning is $1/21 \approx 4.76\%$.
\end{proposition}

\begin{proof}
  Substituting the bin boundaries into the formula given in Proposition \ref{prop:pdist} we find
  $(b - a)(a + b) = 1/(2d + 1)$ for the bin $Bin[e,d], e \in \IZ, d \in \{ 10, \dots, 99 \}$.
  This is minimized for $d = 10$ with a value of $1/21$ as claimed.
\end{proof}

\subsection{Histograms}

Once we have established the binnings, histograms are easy to define.

\begin{definition}\label{def:hist}
  A histogram with domain $D \subset \IR$ is a binning $Bin[i],I$, together with a count function $H: I \ra \IN_{0}$.
\end{definition}

Given a dataset and a binning, we can associate a histogram.

\begin{definition}
  Given binning $Bin[i],I$ of $D \subset \IR$, and a dataset $X = (x_1,\dots,x_n)$ with values in
  $D$, we define the histogram summary of $X$ as the histogram with binning $Bin[i],I$ and count
  function
  \begin{align*}
    H_X(i) = \# \{ j \, | \, x_j \in Bin[i] \, \}.
  \end{align*}
  This means, that $H_X(i)$ counts the number of points of $X$ lying in $Bin[i]$.
\end{definition}

Histograms can be freely merged without loosing information.

\begin{definition}
  Let $H_1, H_2$ be histograms for the same binning. The merged histogram $H_1 + H_2$ has count function:
  \begin{align*}
    (H_1+H_2)(i) = H_1(i) + H_2(i).
  \end{align*}
\end{definition}
The merge operation is clearly associative and commutative.

We can easily see that the histogram merge computes is compatible with merge (concatenation) of datasets:
\begin{proposition}
  Given binning $Bin[i],I$ of $D \subset \IR$, and two datasets $X = (x_1,\dots,x_n),Y=(y_1,...y_m)$ with
  values in $D$. Let $Z=(x_1, \dots, x_n, y_1, \dots, y_m)$ be the merged dataset, then $H_Z = H_X + H_Y$.
\end{proposition}

\begin{proof}
  We have
  \begin{align*}
    H_Z(i) &= \# \{ j \, | \, z_j \in Bin[i] \, \} =
    \# \{ j \leq n \, | \, x_j = z_j \in Bin[i] \, \} +
    \# \{ j > n \, | \, z_j = y_{j-n} \in Bin[i] \, \} \\
    &= H_X(i) + H_Y(i). \qedhere
  \end{align*}
\end{proof}

\def\hyph{{\hbox{-}}}

\begin{remark}
  Let $H$ be a histogram summary of a dataset $X$, and a threshold value $y$.  The count functions
  $count\hyph below_X(y) = \#\{ i | x_i < y \}$ and $count\hyph above_X(y) = \#\{ i | x_i \geq y \}$
  are of great practical interest. If the bin boundaries line-up with the threshold, we can
  get exact approximations of those functions from the histogram.

  In the case of the circllhist, we get exact counts for every 2-digit precision decimal floating point number.

  For logarithmic histograms we get exact counts at the powers of the base $b^e, e\in \IZ$.
\end{remark}

\subsection{Histogram Operations}

Given a histogram we can try to approximate statistics (means, percentiles, etc.) of the original dataset.
There are three strategies for doing so, that we have found valuable in practice:

\begin{enumerate}
\item Derive a probability distribution from a histogram (with uniform distribution inside the
  bins), and apply the statistics to this distribution.
\item Resample the dataset by placing all points inside a bin at equally spaced distances ({\it fair resampling}).
\item Resample the dataset by placing all points inside a bin into the (paretro) midpoints ({\it midpoint resampling}).
\end{enumerate}

From a theoretical perspective the probabilistic strategy is the most natural, and gives generally
good results for data sampled from continues distribution (with density).  Midpoint resampling is
the simplest and gives a low a-priori error bounds independent of the distribution of $X$.  Fair
resampling mitigates between both approaches, and gives lower expected errors on a wide variety of
datasets and while avoiding large relative errors in case of single sample bins.

The current implementation of circllhist uses midpoint resampling to define sum, mean, stddev and
moment estimation. For percentile calculations fair resampling is used.

\begin{proposition}
  Let $X=(x_1,\dots,x_n), n>0$ be a dataset with values in $\IR_{>0}$, let $H$ be the circllhist summary of $X$.
  For each bin $Bin[i]$, let $c_i \in Bin[i]$ be the paretro midpoint. Let
  \begin{align*}
    count[H] = \sum_i H(i), \quad sum[H] = \sum_{i\in I} H(i) \cdot c_i, \qtext{and} mean[H] = sum[H] / count[H].
  \end{align*}
  then
  \begin{enumerate}
  \item $count[H] = count(X) = n$.
  \item The relative error $|sum[H] - sum(X)| / sum(X)$ is smaller than $1/21 \approx 4.76\%$.
  \item The relative error $|mean[H] - mean(X)| / mean(X)$ is smaller than $1/21 \approx 4.76\%$.
  \end{enumerate}
\end{proposition}

\begin{proof}
  The first claim follows immediately from the definition of $count[H]$, and $H(X)$.
  For the second claim, we have
  \begin{align*}
    sum[H] - sum(X) = \sum_{i\in I} H(i) \cdot c_i - \sum_j x_j
    = \sum_{i\in I} ( H(i) \cdot c_i - \sum_{j, x_j \in Bin[i]} x_j)
    = \sum_{i\in I} \sum_{j, x_j \in Bin[i]} (c_i - x_j)
  \end{align*}
  And hence,
  \begin{align*}
    |sum[H] - sum(X)| \leq \sum_{i,j, x_j \in Bin[i]} |c_i - x_j| \leq \sum_{j} \frac{1}{21} |x_j| = \frac{1}{21} \cdot sum(X),
  \end{align*}
  where we used \ref{prop:21} in second step.
  This shows the second claim.
  For the third claim follows from the second by extending the fraction with $1/n$.
\end{proof}

\subsection{Quantiles}
\label{sec:quantiles}

Quantiles can be approximated from histograms using any of the strategies outlined at the beginning
of the last section.  The probabilistic strategy has the downside, that for $q$ close to $1$, the
$q$-quantile will always be near the high end of the largest bin with at least one sample in
it. Often times this bin contains only a single sample. In this case the worst case error of the
full bin size can be easily assumed.

Midpoint resampling has the downside, that expected errors for percentiles in the center of the
distribution are much larger than needed. For densely populated bins the uniform distribution is
often a good approximation, and can be used to estimate percentiles with high precision in typical
cases.

Fair resampling offers a welcome middle route. Instead of placing all samples at the (paretro)
midpoint, or smoothing them out with a uniform distribution, we place the samples at equally
spaced position within each bin. In this way, densely populated bins will be approximated by a near
uniform distribution and bins with only a single sample will be replaced by a single sample at the
midpoint, reducing the worst-case error to half the bin size.

\begin{definition}
  Given a histogram $H$ on a binning $Bin[i], i \in I$ we define the fair resampling $X$ of $H$ as follows.
  For each bin $Bin[i]$, with boundaries $a,b$ and count $H(i)=n$, we consider the points
  \begin{align*}
    x_{i,k} = a + \frac{k}{n+1} (b-a) \qtext{for} k=1,\dots,n
  \end{align*}
  and let $\hat{X} = (x_{i,k} \,|\, i \in I, k = 1\dots,H(i))$. This is a dataset with $count[H]$ points.

  Similarly, we define the midpoint resampling of $H$, with $H(i)$ samples at the midpoint of $Bin[i]$,
  and the paretro midpoint resampling of $H$, with $H(i)$ samples at the paretro midpoint of $Bin[i]$.
\end{definition}

We now want to define the quantiles of a histogram as quantiles of the fair resampling of $H$.
Before we can do so we first need to clarify the quantile definition for datasets.
While there is only a single established quantile definition for probability distributions, there
are multiple different quantile definitions for datasets found in the wild.
A comprehensive list was compiled by Hyndman-Fan in 1996 \cite{HF1996}.
For our discussion we will need type-1/7 quantiles from the Hyndman-Fan list,
as well as two other quantile functions, used by data-structures covered in our evaluation.

\begin{definition} \label{def:quantiles}
  Given a dataset $X=(x_1,\dots,x_n), n>0$ of real numbers, and a number $q \in [0,1]$.
  Let $x_{(1)} \leq x_{(2)} \leq \dots \leq x_{(n)}$ be the ordered version of $X$.
  We define the minimal type-1 q-quantile as
  \begin{align*}
    Q^1_0(X) = x_{(1)} \qtext{and}  Q^1_q(X) = x_{(\ceil{q \cdot n})} \qtext{for} q > 0.
  \end{align*}
  We define the minimal type-7 quantile as
  \begin{align*}
    Q^{7}_q(X) = x_{(\floor{q \cdot (n - 1)} + 1)}.
  \end{align*}
  The interpolated type-7 quantile is given by:
  \begin{align*}
    Q^{7i}_q(X) = (1 - \gamma) \cdot x_{(\floor{q \cdot (n - 1)} + 1)} + \gamma \cdot x_{(\ceil{q \cdot (n - 1)} + 1)},
  \end{align*}
  where the interpolation factor $\gamma \in [0,1]$ is given by $\gamma = q(n-1) - \floor{q(n-1)}$.

  The type-hdr quantile $Q^{hdr}_q(X)$ is given by $x_{(1)}$ if $qn \leq \frac{1}{2}$, by $x_{(n)}$
  if $qn \geq n - \frac{1}{2}$, and otherwise, by
  \begin{align*}
    Q^{hdr}_q(X) &= x_{(\floor{q n - \frac{1}{2}})}.
  \end{align*}

  Similarly, the type-tdigest quantile is given by $x_{(1)}$ if $qn \leq \frac{1}{2}$, by $x_{(n)}$
  if $qn \geq n - \frac{1}{2}$, and otherwise, by
  \begin{align*}
    Q^{tdigest}_q(X) &= (1 - \gamma) \cdot x_{(\floor{q n - \frac{1}{2}})} + \gamma \cdot x_{(\ceil{q n - \frac{1}{2}})}
  \end{align*}
  where $\gamma = q n - \frac{1}{2} - \floor{q n - \frac{1}{2}}$.
\end{definition}

\begin{figure}
  \begin{subfigure}{0.5\textwidth}
    \includegraphics[width=\textwidth]{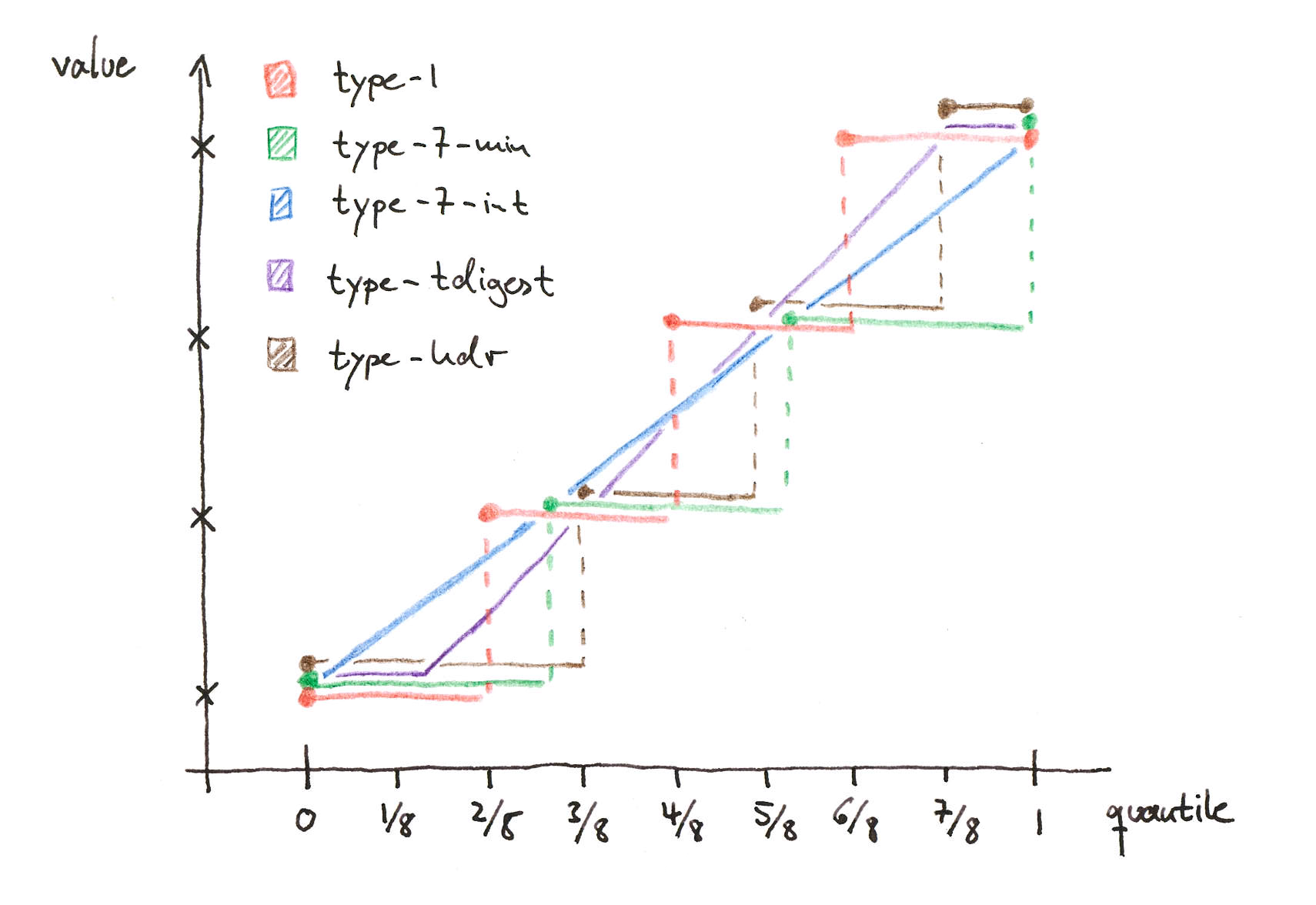}
    \caption{Theoretical Quantile Functions}
    \label{fig:tqf}
  \end{subfigure}
  \begin{subfigure}{0.5 \textwidth}
    \includegraphics[width=\textwidth]{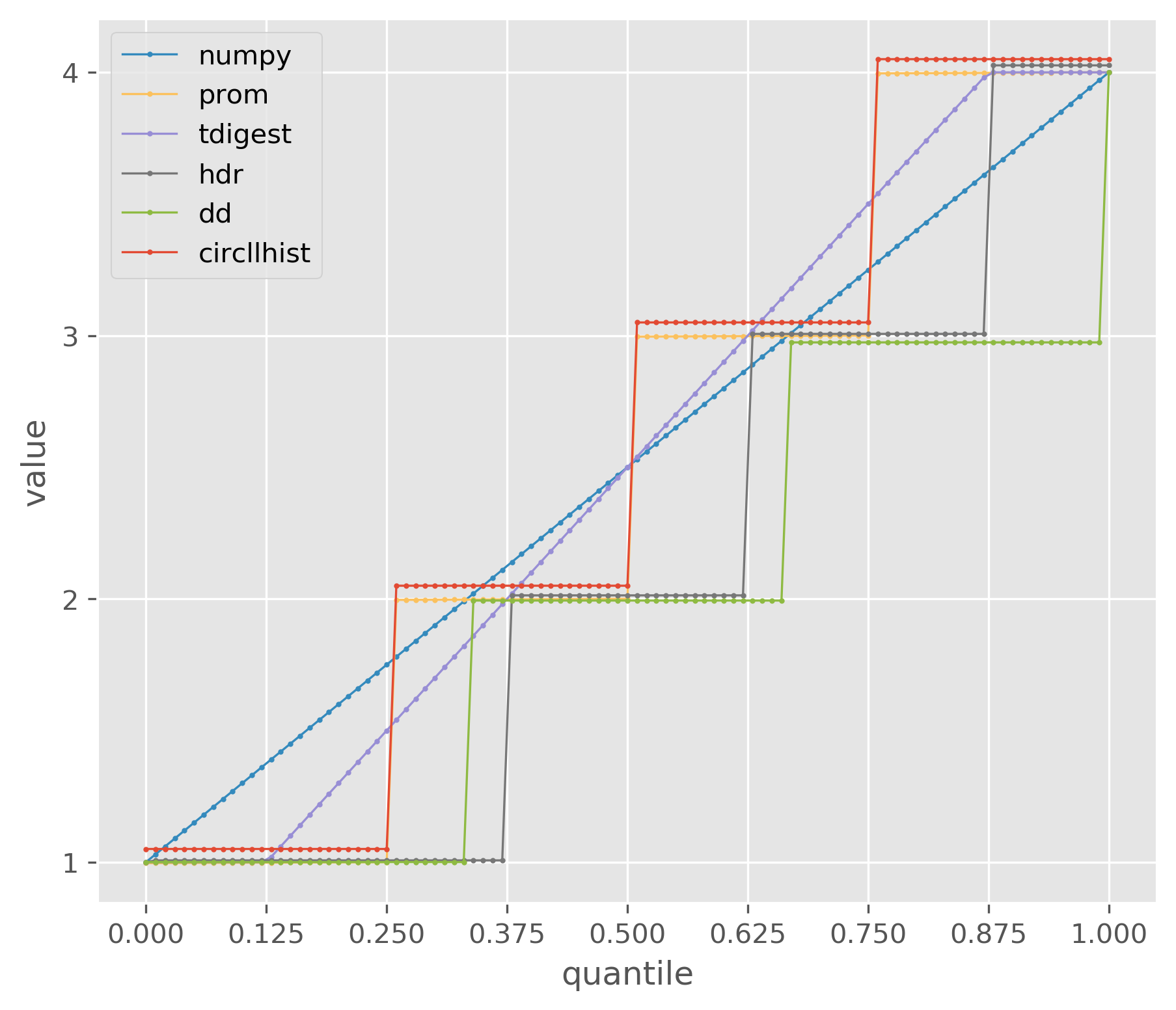}
    \caption{Computed Quantile Functions}
    \label{fig:pqf}
  \end{subfigure}
  \caption{Quantile Function Comparison for the Dataset $[1,2,3,4]$}
  \label{fig:qf}
\end{figure}

Figure \ref{fig:tqf} shows a plot of these quantile definitions as functions on $q$.  From a
theoretical perspective type-1 quantiles are the most natural, since they correspond the
probabilistic quantiles for the empirical distribution function of the dataset $X$.  Also they allow
to formulate precise statements about lower counts (cf. Proposition \ref{prop:count} below).  For
this reason the circllhist implementation uses minimal type-1 quantiles on the fair resampling of
the histogram as for quantile calculations.

Somewhat surprisingly, type-7 quantiles are most commonly found in software packages.
E.g. numpy \cite{numpy} computes interpolated type-7 quantiles by default.
Also DDSketch \cite{dd} and t-digest \cite{tdigest} implement variants of type-7 quantiles.
A detailed comparison between type-1 and type-7 quantiles can be found at \cite{HH19}.
The t-digest and the HDR Histograms use custom quantile functions that were not covered in Hyndman-Fan.
We call them ``type-hdr'' and ``type-tdigest'' here.
Figure \ref{fig:pqf} shows a plot of the quantile functions, produced by all relevant implementations.

\begin{proposition}\label{prop:count}
  Let $X$ be a dataset of length $n > 1$, $y \in \IR$ a threshold value and $0 < q \leq 1$.
  Then the following statements are equivalent:
  \begin{itemize}
  \item[(1)] There are at least $q n$ datapoints $x$ in $X$ with $x \leq y$.
  \item[(2)] We have $Q^1_q(X) \leq y$.
  \end{itemize}
\end{proposition}

\begin{proof}
  To show (1) implies (2), assume that there are at least $q n$ datapoints below $y$.
  Since the number of datapoints is always an integer, the same holds for $r = \ceil{q n}$.
  This implies that the ordered values $x_{(1)}, \dots, x_{(r)}$ all lie below $y$,
  and hence $x_{(r)} = Q^1_q(X) \leq y$.

  To show (1) implies (2), assume that $Q^1_q(X) \leq y$ then the $r = \ceil{n q}$ values $x_{(1)},
  \dots, x_{(r)} = Q^1_q(X)$ all lie below $y$. This shows there are at least $n q$
  datapoints lie below $y$.
\end{proof}

\begin{proposition}
  Let $X$ a dataset with $n > 0$ positive values and $0 < q \leq 1$. Let $H$ be the circllhist summary of $X$.
  We denote by $\hat{X}^f(H)$ the fair resampling, and by $\hat{X}^p(H)$ the paretro-midpoint resampling of $H$.
  \begin{enumerate}
  \item The relative error of the estimated type-1 quantile is less than $1/21 \approx 4.76\%$ for
    the paretro-resampling:
    \begin{align*}
      |Q^1_q(X) - Q^1_q(\hat{X}^p(H))| \leq \frac{1}{21} Q^1_q(X)
    \end{align*}
  \item The relative error of the estimated type-1 quantile is less than $1/10 = 10\%$ for
    the fair-resampling:
    \begin{align*}
      |Q^1_q(X) - Q^1_q(\hat{X}^f(H))| \leq \frac{1}{10} Q^1_q(X)
    \end{align*}
    Furthermore, in case the quantile value falls into a bin with a single sample (typical for
    outliers), the maximal relative error is $5\%$.
  \end{enumerate}
\end{proposition}

\begin{proof}
  Let $r = \ceil{q n}$, so that $Q^1_q(X) = X_{(r)}$.
  By construction of the paretro resampling, the $r$-th ordered value in the $Q^1_q(\hat{X}^p) = \hat{X}^p_{(r)}$,
  is the paretro midpoint of the bin containing $X_{(r)}$.
  This shows that $|Q^1_q(X) - Q^1_q(\hat{X}^p)| < 1/21 Q^1_q(X)$ by Proposition \ref{prop:21}.

  Similarly for the fair resampling, the $r$-th ordered value in the $Q^1_q(\hat{X}^f) = \hat{X}^f_{(r)}$,
  lies in the bin containing $X_{(r)}$.
  Now we claim, that the maximal relative distance to any point in the same circllhist bin is $10\%$.
  Indeed, the worst case relative difference is assumed for a bin $[10 \cdot 10^e, 11 \cdot 10^e)$ with
  $x = 10 \cdot 10^e$ and $y \ra 11 \cdot 10^e$ so that $|x-y|/x \ra 1/10 = 10\%$.

  In case that there is only a single sample in the bin $[a,b)$ containing $X_{(r)}$, the value
  $\hat{X}^f_{(r)}$ will be at the (arithmetic) midpoint of the bin $m = a+\frac{1}{2}(b-a)$.
  Hence, in this case $|X_{(r)} - \hat{X}^f_{(r)}| / X_{(r)} \leq \frac{1}{2} (b-a)/X_{(r)}$
  This number is maximized for $X_{(r)} = a$, and $a = 10 \cdot 10^e$ in the circllhist binning.
  In which case $\frac{1}{2} (b-a)/a =\frac{1}{2} \frac{1}{10} = 5\%$.
\end{proof}

\begin{example}
  With the notation from the last proposition.  Let $X=(10,10,10, \dots, 10)$ of length $n$, then
  the fair resampling is given by $\hat{X}^f = (10 + \frac{1}{n+1}, \dots, 10 + \frac{n}{n+1})$.
  So $Q^1_1(\hat{X}^f) = 10 + \frac{n}{n+1}$ which converges to 11 for $n \ra \infty$.
  On the other hand $Q^1_1(X) = max(X) = 10$. So the worst case error of 10\% is assumed in
  the asymptotic case.
\end{example}

This proposition shows, that the worst-case relative error for the circllhist is 10\%.  The
theoretical worst-case is only realized in cases were a large number of samples falls at the lower
end of a bin as shown in the last example.  This is very rare to happen.  In practice we usually
see a 5\% worst-case relative error at the tails of the distribution, and high accuracy
quantiles (<\%1 relative error) in the body of the distribution.

\section{The Circllhist Implementation}

The data-structure implemented in \cite{libcircllhist}, is a circllhist in the sense of
Definition~\ref{def:circllhist} with exponent range limited to $-128 \leq e < 127$.
In this way, both exponent $e$ and mantissa $d$ can be represented by 8-bit integers.
The counts $H(i)$ are represented as 64-bit unsigned integers.

The largest representable number of the circllhist is $99 \cdot 10^{127}$.
This number is larger than the age of the universe measured in nano-seconds (13.8 billion years)
The smallest representable positive number is $10 \cdot 10^{-128}$.
This number is smaller than the Plank time measured in years ($5.39 \cdot 10^{-44}$ s).
We have found this value range to be sufficient for all practical purposes.

Within that range the circllhist bins have a relative size of 1\%-10\% of the values contained in
them. Hence, the original values can be reconstructed with a relative error of no more than 10\%.
If the reconstructed values are placed at the paretro midpoint of the bin, the relative reconstruction
error can be reduced to 4.76\% (see Proposition \ref{prop:21}).

The accuracy of approximations of statistical quantities like sums, means and quantiles is commonly
much better than 10\%, since the individual reconstruction errors cancel across the dataset.
In the case of quantiles, we will see this in the evaluation below.

The counts of samples above or below a threshold, can be reconstructed accurately for threshold
that fall onto bin boundaries. In the case of the circllhist, those lie at decimal floating point
numbers with two digits of precision (e.g. 0.23, 1.5, 110). In practice, those are the numbers
that humans choose, if they have to come up with thresholds.

A circllhist is internally represented as a sparse list of (bin, count) pairs.  If a bin has a count
of 0, then it can be skipped in the representation.  The serialized form of the histogram will only
contain bins with non-zero counts.  The theoretical maximum of used bins is $2 \times 256 \times 90
+ 1 = 46081$ (for sign, exponent, mantissa and zero bucket) which makes for a maximal size of 461kb.

In practice we have never seen histogram data structures getting close to this size.  Even in
extreme cases, when capturing billions of samples from a nano-second to minute scale, the number of
allocated bins never exceeded 1000, and the total size stayed below 10kb.
Typical histograms in our system, have anywhere from 0-200 allocated bins and occupy <2kb before
compression.

A notable design goal of the circllhist is it's use for measurements inside the kernel or
low-powered embedded devices.  In those environments floating point arithmetic is not available, and
insertion performance is particularly critical. For these purposes the circllhist provides a highly
optimized insertion function that avoids floating point arithmetic entirely (cf. Proposition
\ref{prop:rec}).

The C implementation of libcircllhist includes a number of performance optimizations.  It comes with
an optional index structure, that avoids iteration over bins when retrieving and inserting data and
uses static branch annotations to aid CPU branch predictions.  With these optimizations we can get
raw insertion latencies down to $\sim 10ns$ for integer, and $\sim 80ns$ for double
values.\footnote{ These latencies were measured in a tight C loop with the provided
\texttt{test/histogram\_perf.c} script, on a 2Ghz Intel Xeon CPU.  The evaluation in
section~\ref{sec:eval} is Python based and uses different iteration counts and data. }

Implementations of the circllhist are available for a variety of languages including:
\begin{itemize}
\item C/C++ -- \url{https://github.com/circonus-labs/libcircllhist}
\item Python -- \url{https://github.com/circonus-labs/libcircllhist/tree/master/src/python}
\item Go -- \url{https://github.com/circonus-labs/circonusllhist}
\item Lua -- \url{https://github.com/circonus-labs/libcircllhist/tree/master/src/lua}
\item JavaScript -- \url{https://github.com/circonus-labs/circllhist.js}
\item C\# .NET -- \url{https://github.com/circonus-labs/netcircllhist}.
\end{itemize}

\section{Evaluation}
\label{sec:eval}

\begin{figure}
   \includegraphics[width=\textwidth/2]{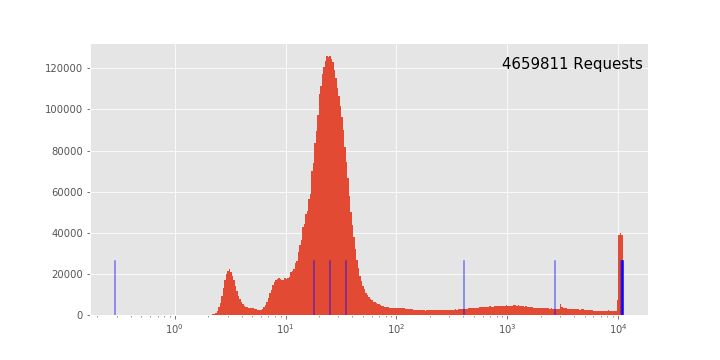}
   \includegraphics[width=\textwidth/2]{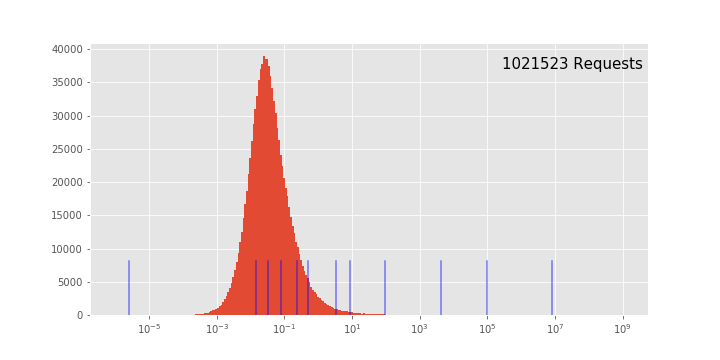}
   \caption{Total Distribution of the 'API Latency' and 'Simulated Latencies' Datasets, with quantile markers.}
   \label{fig:ds}
\end{figure}

In this section we present a numerical comparison between the quantile aggregation methods described
in Section \ref{sec:rw}. We will evaluate the precision and performance of quantile calculations,
performance of insertion and merge operations, as well as the size of the resulting data-structures.

The evaluation proceeds in three phases: insertion, merge and quantile calculation. In the insertion
phase, raw data is inserted the data-structures. Each data-set is split into a number of batches. We
create individual data-structures for each batch. In the merge phase, data-structures created for
each batch are merged into a single one.  In the quantile phase, we perform quantile calculations on
the merged data-structure.

The evaluation was performed using a set of Jupyter notebooks, which are available at
\begin{center}
  \url{https://github.com/circonus-labs/circllhist-paper}.
\end{center}
This repository contains datasets and source code used to generate the exact tables and graphics we
are using this document.  We also provide docker images and instructions, that should aid the
reproduction of the evaluation results on other machines.  We are open to extension and improvements
to the evaluation setup. Please contact us via email or open a pull request, if you identified flaws
or found improvements.

\subsection{Methods}

The following data-aggregation methods are considered for evaluation.

\begin{itemize}
\setlength{\itemindent}{2em}
\item[exact] Exact quantile computation based on numpy arrays \cite{numpy}.
\item[prom] Quantile estimation based on Prometheus Histograms \cite{prom}.
\item[hdr] The HDR Hisotgram data-structure introduced in \cite{hdr}.
\item[dd] The DDSketch data-structure introduced in \cite{dd}.
\item[t-digest] The t-digest data-structure introduced in \cite{tdigest}.
\item[circllhist] The circllhist data-structure described in this document.
\end{itemize}

The Prometheus monitoring system provides a histogram data-type that consists of a list of ``less
than'' metrics, which count how many samples were inserted that are below manually configured
threshold values. This datum of a Prometheus histogram is equivalent to a histogram in the sense of
Definition \ref{def:hist}, with bin-boundaries at the configured thresholds.  Prometheus Histograms
are included here for their wide use in practice.  The method itself is not really comparable, since
it the results are highly dependent on the number and location of the chosen thresholds. For the
evaluation, we follow the recommended practice of using a total ten bins at locations which cover
the whole data range with emphasis on the likely quantile locations.  We use a hand-written Python
translation of the original quantile functions written in go.

The HDR Histogram data structure is a log-linear Histogram in the sense of Definition \ref{def:ll}.
It uses a base of 2, and a configurable range and precision.  It's notable that the exposed API
let's the user specify decimal precision.  The internal precision is then chosen in such a way that
the resulting base-2 bins are smaller than the specified base-10 accuracy.  Using base-2 arithmetic
has the advantage of allowing bit-wise manipulation of floating points numbers to determine the bin
location. For our evaluation used the Python
implementation\footnote{\url{https://github.com/HdrHistogram/HdrHistogram_py}}, and configured the
HDR Histograms to cover the same range of the circllhist ($10^{-128}\dots10^{+127}$) with two digits
of decimal precision.

The DDSketch is a histogram for the logarithmic binning introduced in Example \ref{example:log}.
It allows arbitrary positive real numbers as basis, to configure the desired precision.
We use the Python implementation\footnote{\url{github.com/DataDog/sketches-py}}, with the default
precision of 1\% (corresponding to $b=101/99$).

The t-digest is the only evaluated method, that is not internally using a histogram data structure.
Instead the inserted data is aggregated into clusters of adaptive size.  The sizes are chosen in
such a way that, high resolution data is available at the tails of the distribution.  We use the
original Java implementation\footnote{\url{github.com/tdunning/t-digest}} called from python
using pyjnius\footnote{\url{https://github.com/kivy/pyjnius}}. The digests are configured with a
compression parameter of 100, which is described to be a ``common value for normal uses'' in
the source code.

For the circllhist we used the Python binding\footnote{\url{github.com/circonus-labs/libcircllhist}}.
It does not allow any configuration of bin sizes or accuracy.

\subsection{Datasets}

For the evaluation we choose three different datasets: ``Uniform Distribution'', ``API Latencies'' and ``Simulated API Latencies''
each containing over 1 million samples, split into more than 1000 batches.

For the ``Uniform Distribution'' dataset, random samples are generated for a uniform distribution on $[10,100]$.
These are split in to 1000 batches each containing 100 values, for a total of 100.000 samples.

The ``API Latencies'' dataset contains data collected at one of our internal APIs collected in 10
minute batches over the course of 6 weeks.  The data is spread between 0.29 and 11.000 (ms).  There
is a cutoff value around 11sec, that was caused by timeout logic getting triggered.  This results in
the higher quantile values being close together.  It consists of a total of 4.65 million samples in
6048 batches.

For the ``Simulated Latencies'' dataset, we generate a total of 1000 batches of randomized sizes
with randomized data.  The size of the batches follows a geometric distribution.  The samples
themselves are generated as a randomly displaced and scaled paretro distribution.  The data is
spread between on between $10^{-5}$ and $10^{10}$ with an extremely long tail.  The dataset contains
a total of close to 1.000.000 samples.

Figure \ref{fig:ds} contains a histogram visualization of the total distribution of these datasets.

\subsection{Size}

\begin{figure}
    \begin{subfigure}{0.33\textwidth}
      \includegraphics[width=\textwidth]{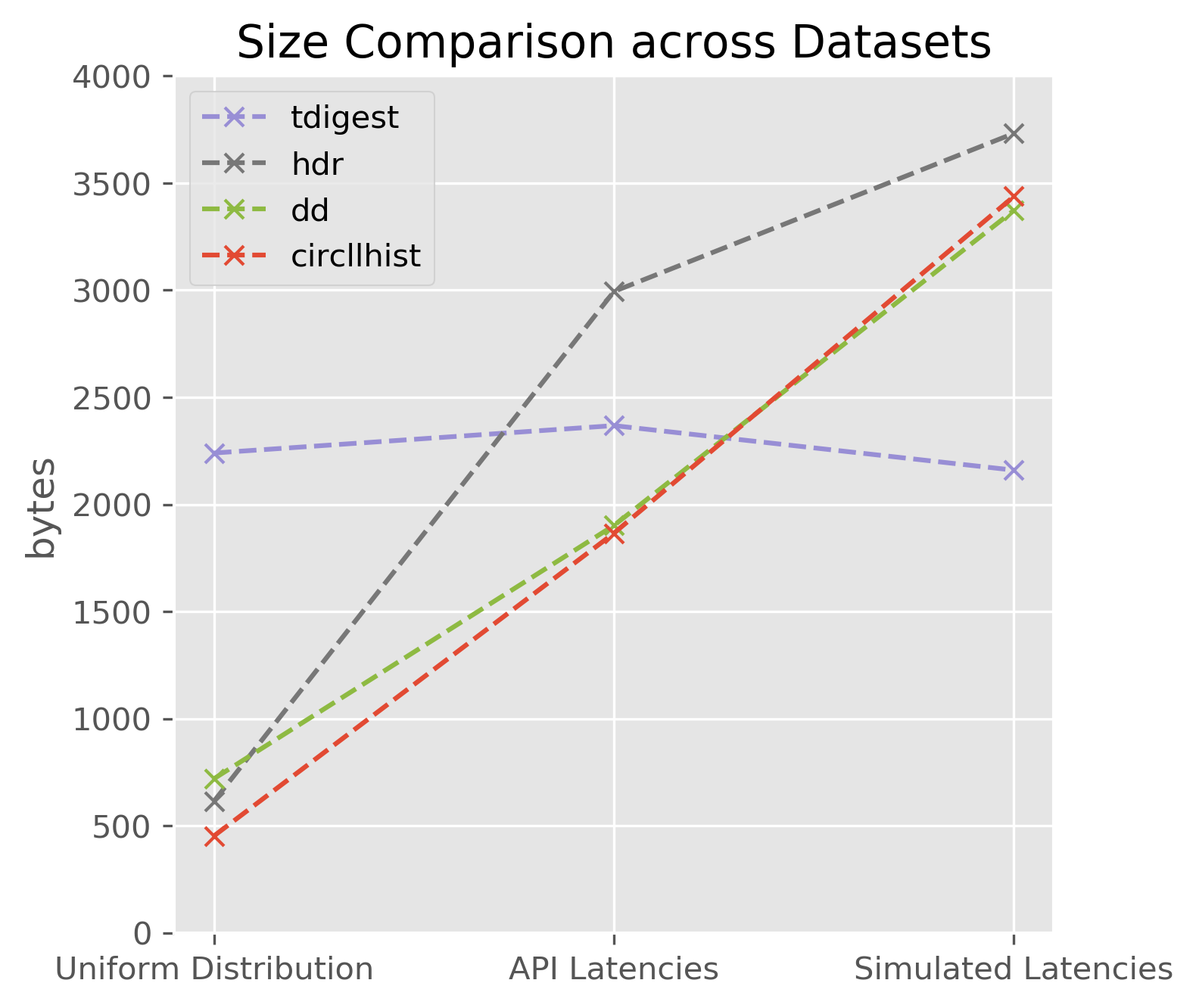}
      \caption{Size Comparison}
      \label{fig:size}
    \end{subfigure}
    \begin{subfigure}{0.66\textwidth}
      \begin{tabular}{lrrr}
\toprule
{} &  Uniform Distribution &  API Latencies &  Simulated Latencies \\
\midrule
exact      &                800000 &       37278488 &              8172184 \\
prom       &                    88 &             88 &                   88 \\
tdigest    &                  2240 &           2368 &                 2160 \\
hdr        &                   615 &           2994 &                 3732 \\
dd         &                   720 &           1902 &                 3373 \\
circllhist &                   453 &           1866 &                 3438 \\
\bottomrule
\end{tabular}

      \caption{Tabluated sizes in bytes}
      \label{fig:tsize}
    \end{subfigure}
    \caption{Size Comparison}
\end{figure}

Figure \ref{fig:size} and Table \ref{fig:tsize} show the sizes of the data structures after all samples
from the respective datasets have been inserted and the data structures have been merged.

With our choice of parameters, the size of the HDR Histogram, DDSketch, circllhist follow each other
quite closely. Note that the size increases with the spread of the data, not with the size of the
inserted samples.

It's important to note, that we compare the sizes of the serialized versions of the data-structures,
as they would be consumed on disk or on the network, not the in-memory size (as this is harder to
estimate). The in-memory size might be larger than the serialized size. This is particularly true
for the HDR Histogram with pre-allocates all bins in memory, and skips empty bins for the
serialization.

The Prometheus histogram stores a count for each of the ten configured threshold value, as well as
the total count (infinite bin), and hence consumes exactly 88 bytes.

The size of the t-digest is relatively constant across all three data-sets, which reflects the fact
that the number of clusters is kept constant when additional data is inserted.

\subsection{Performance}

\begin{figure}
    \begin{subfigure}{\textwidth}
      \includegraphics[width=\textwidth/3]{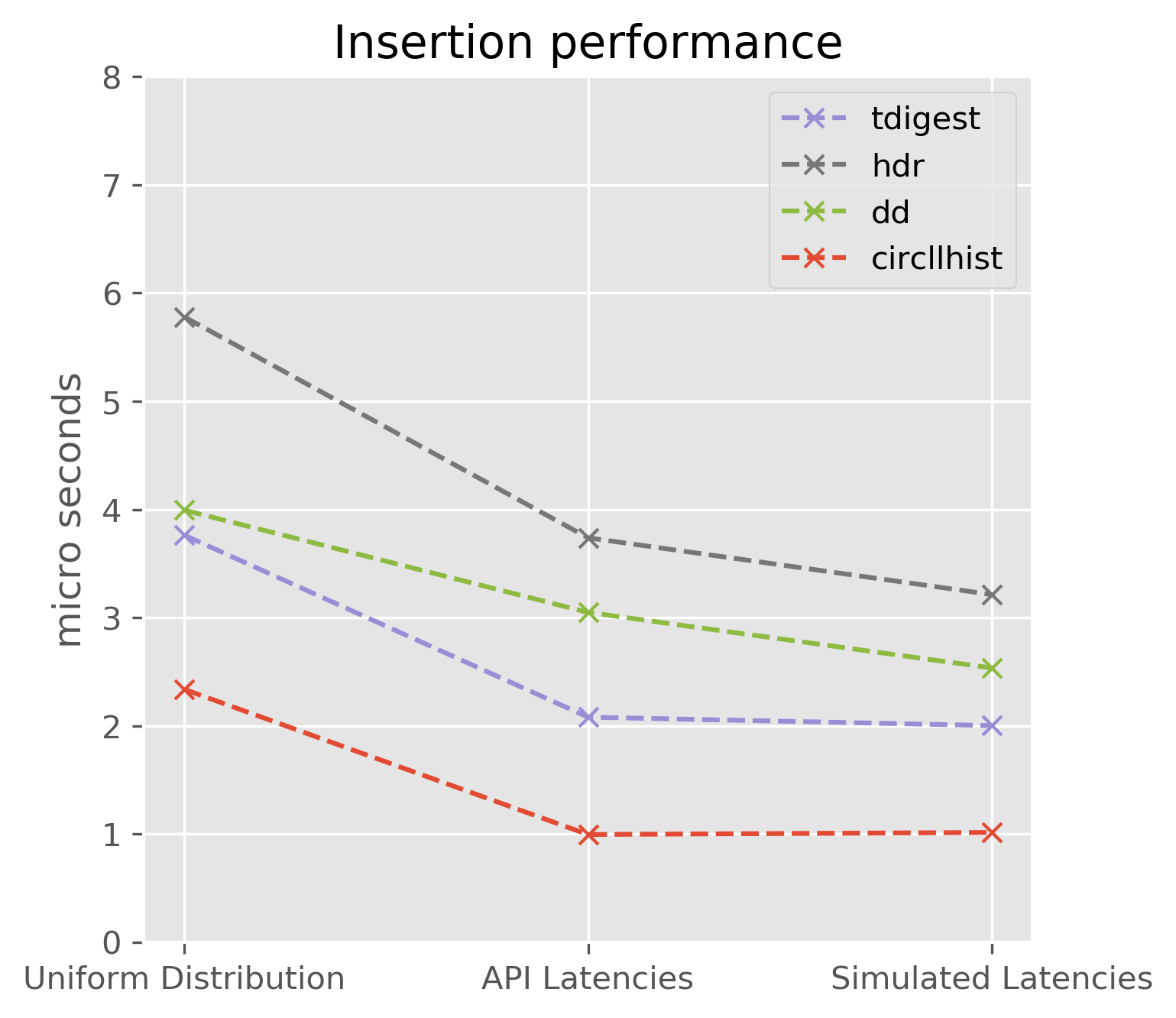}
      \includegraphics[width=\textwidth/3]{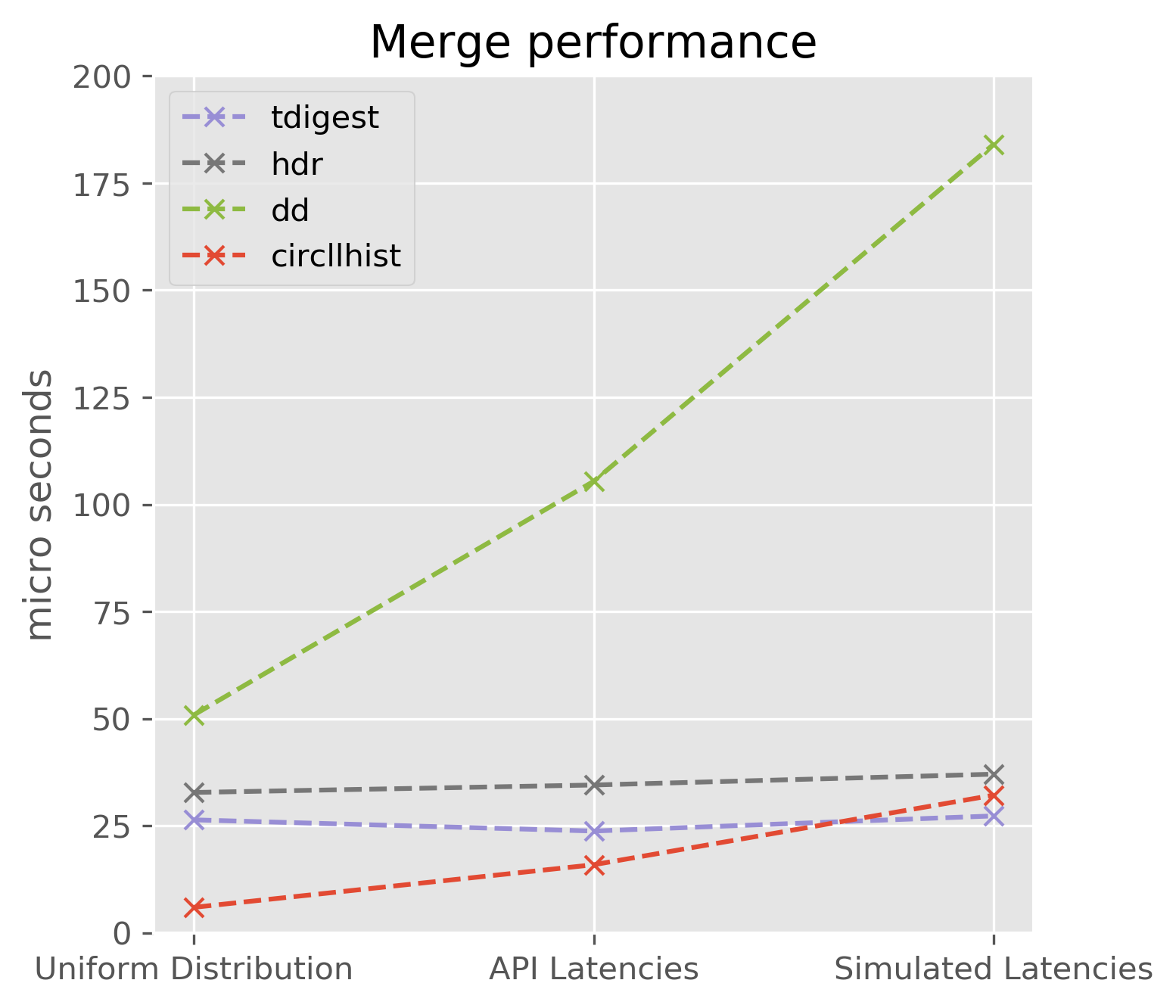}
      \includegraphics[width=\textwidth/3]{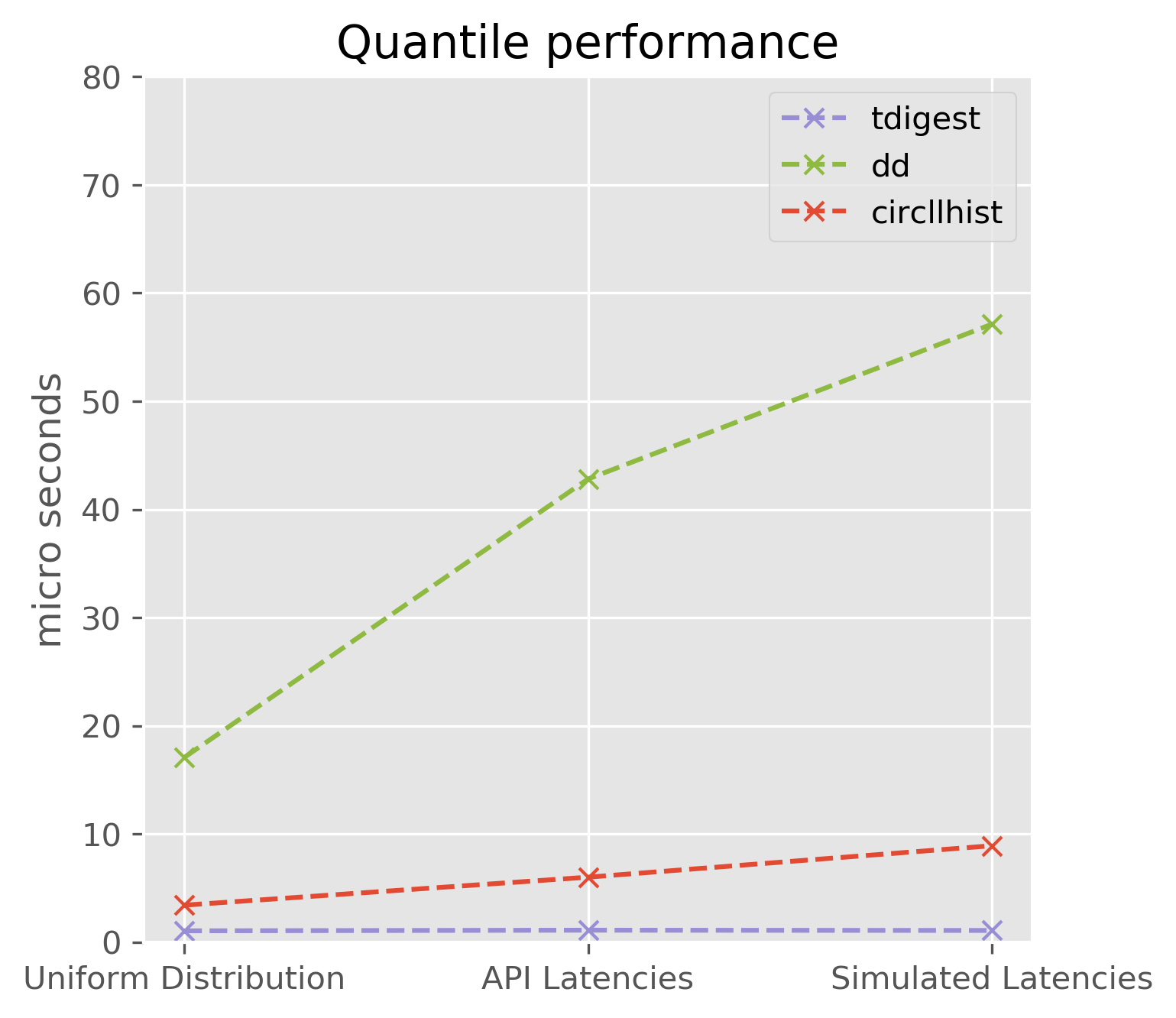}
      \caption{Performance Comparison}
      \label{fig:perf}
    \end{subfigure}
    \begin{subfigure}{\textwidth}
      \centering
      \begin{tabular}{lrrrrrrrrr}
\toprule
Phase & \multicolumn{3}{l}{Insertion} & \multicolumn{3}{l}{Merge} & \multicolumn{3}{l}{Quantile} \\
Dataset &  Unif. D. & API L. & Sim. L. & Unif. D. & API L. & Sim. L. & Unif. D. &  API L. & Sim. L. \\
\midrule
exact      &       1.3 &    0.0 &     0.1 &     33.0 & 3351.9 &   472.6 &    825.1 & 25199.5 &  8789.9 \\
prom       &       8.9 &    7.2 &     6.6 &      0.9 &    0.9 &     0.9 &     10.0 &     8.6 &    10.2 \\
tdigest    &       3.8 &    2.1 &     2.0 &     26.4 &   23.8 &    27.3 &      1.1 &     1.1 &     1.1 \\
hdr        &       5.8 &    3.7 &     3.2 &     32.8 &   34.5 &    37.1 &   1384.3 &  1981.2 &  1604.7 \\
dd         &       4.0 &    3.0 &     2.5 &     50.9 &  105.4 &   184.0 &     17.1 &    42.8 &    57.1 \\
circllhist &       2.3 &    1.0 &     1.0 &      6.0 &   15.9 &    32.2 &      3.4 &     6.0 &     8.9 \\
\bottomrule
\end{tabular}

      \caption{Tabulated performance data in usec\protect\footnotemark}
      \label{tab:perf}
    \end{subfigure}
    \caption{Performance Comparison}
\end{figure}

\footnotetext{
  The insertion time is reported as per inserted sample.
  The merge time is the time per merged batch.
  The quantile time is reported per calculated quantile.
}

Figure \ref{fig:perf} and Table \ref{tab:perf} show the measured performance for insertion, merge
and analysis phase for the three considered datasets. As always, performance measurements should be
taken with a grain of salt, since they are heavily influenced by the implementation, configuration
and hardware choices. In our case, we choose to perform all performance measurements in python, with
the most popular python implementations available.  In the case of the t-digest we ran into accuracy
(and performance) problems with the the python version so we used the official Java version via
pyjnius.

The measurements itself were performed with the
timeit\footnote{https://docs.python.org/3/library/timeit.html} package, with 5 consecutive batches
of runs each taking longer than 0.2 seconds.  As recommended, we report the minimal duration that
was attained in any of the runs.

We ran these experiments on a server with Intel(R) Xeon(R) D-1540 CPU and 128GB RAM running
Linux 5.3.13 and python 3.7.3.

With this configuration, the circllhist was consistently the fastest method when it comes to
insertion and merge.  For the quantile calculations the t-digest is very efficient, followed by
circllhist and DDSketch. Quantile calculations for the HDR Histogram take a lot longer, which
probably indicates optimization potential.

\subsection{Accuracy}

As we have seen in section \ref{sec:quantiles}, there are a number of different quantile definitions
circulating in the wild. When evaluating quantile accuracy, we have to make sure we are comparing
the computed numbers to the theoretical quantiles the methods are approximating.  With the notation
of Definition~\ref{def:quantiles}, the circllhist and Prometheus approximate type-1 quantiles,
DDSketch approximates minimal type-7 quantiles.
t-digest and HDR Histograms approximate custom quantile functions, that we called ``type-tdigest''
and ``type-hdr'' in Definition~\ref{def:quantiles}.

In all cases, the following quantile values were considered:
\begin{align*}
 0, 0.25, 0.5, 0.75, 0.9, 0.95, 0.99, 0.995, 0.999, 0.9999, 0.99999, 1.
\end{align*}
For each value, the quantile is computed once with the evaluated data-structure and once with the
respective theoretical function on the raw data. The relative difference between those value is
reported in Figure \ref{fig:acc} and Table \ref{tab:acc}.

\begin{figure}
  \begin{subfigure}{\textwidth}
    \includegraphics[width=\textwidth]{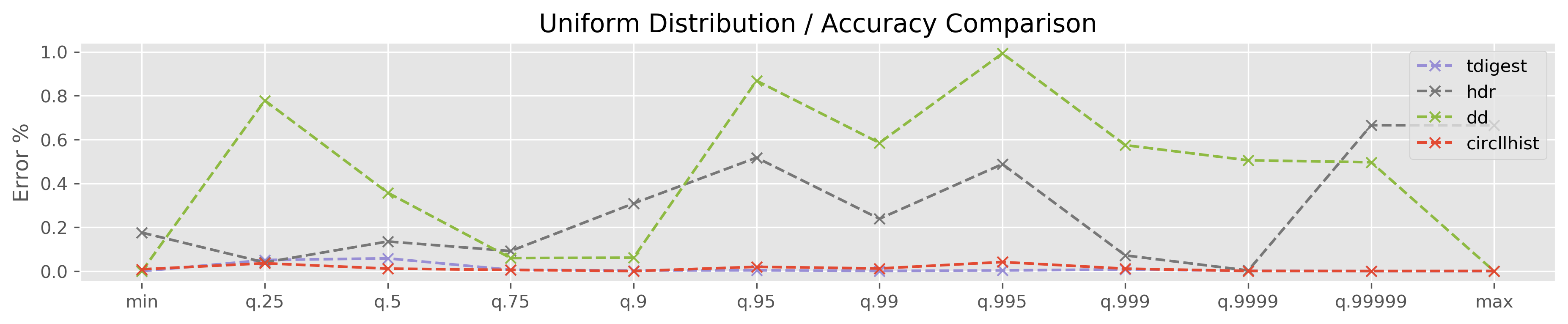}
    \includegraphics[width=\textwidth]{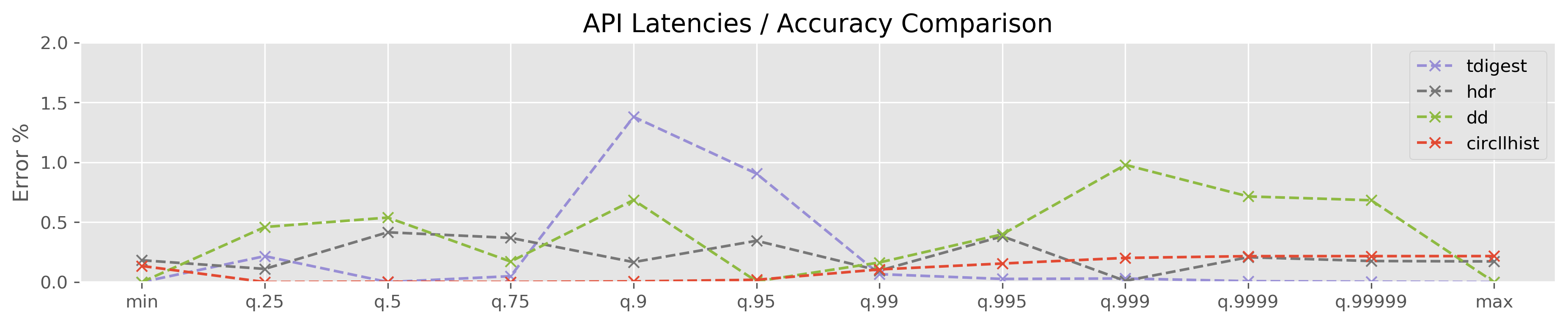}
    \includegraphics[width=\textwidth]{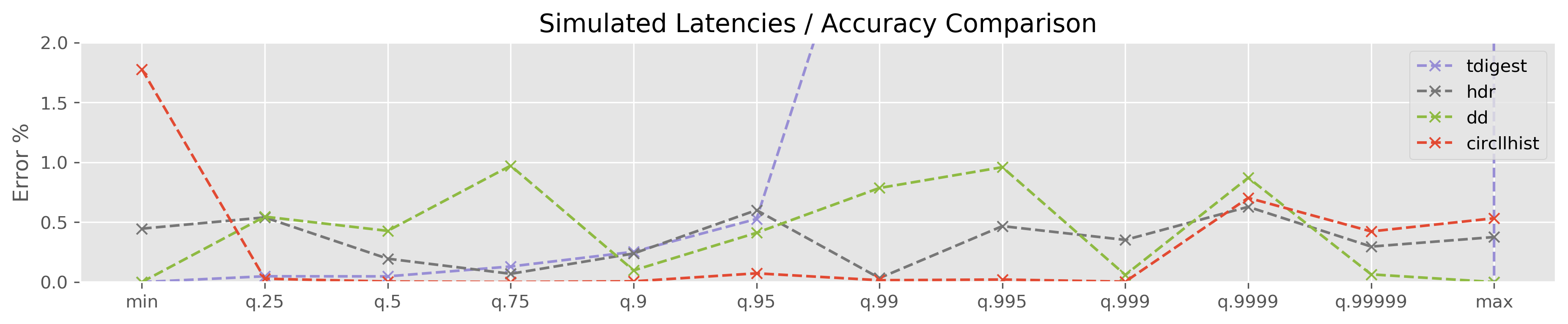}
    \caption{Accuracy Comparison}
    \label{fig:acc}
  \end{subfigure}
  \begin{subfigure}{\textwidth}
    \begin{minipage}{\textwidth}
      \scriptsize
      \begin{tabular}{llrrrrrrrrrrrr}
\toprule
                    & Quantile &      0 &    0.25 &     0.5 &   0.75 &    0.9 &  0.95 &   0.99 &  0.995 &  0.999 &  0.9999 &  0.99999 &     1 \\
\midrule
Uniform Distribution & prom & 100.00 &    0.01 &    0.00 &   0.01 &   0.06 &  0.01 &   0.01 &   0.04 &   0.01 &    0.00 &     0.00 &  0.00 \\
                    & tdigest &   0.00 &    0.05 &    0.06 &   0.01 &   0.00 &  0.00 &   0.00 &   0.00 &   0.01 &    0.00 &     0.00 &  0.00 \\
                    & hdr &   0.18 &    0.04 &    0.14 &   0.09 &   0.31 &  0.52 &   0.24 &   0.49 &   0.07 &    0.00 &     0.67 &  0.67 \\
                    & dd &   0.00 &    0.78 &    0.36 &   0.06 &   0.06 &  0.87 &   0.59 &   0.99 &   0.57 &    0.51 &     0.50 &  0.00 \\
                    & circllhist &   0.01 &    0.04 &    0.01 &   0.01 &   0.00 &  0.02 &   0.01 &   0.04 &   0.01 &    0.00 &     0.00 &  0.00 \\
API Latencies & prom & 100.00 &   52.67 &  126.47 & 145.16 &  38.95 & 10.63 &   5.66 &   7.31 &   8.58 &    8.86 &     8.89 &  8.89 \\
                    & tdigest &   0.00 &    0.22 &    0.00 &   0.05 &   1.38 &  0.91 &   0.07 &   0.03 &   0.03 &    0.01 &     0.00 &  0.00 \\
                    & hdr &   0.18 &    0.11 &    0.42 &   0.37 &   0.17 &  0.34 &   0.10 &   0.38 &   0.01 &    0.21 &     0.18 &  0.17 \\
                    & dd &   0.00 &    0.46 &    0.54 &   0.17 &   0.69 &  0.01 &   0.16 &   0.40 &   0.98 &    0.71 &     0.68 &  0.00 \\
                    & circllhist &   0.13 &    0.00 &    0.00 &   0.00 &   0.01 &  0.02 &   0.11 &   0.15 &   0.20 &    0.22 &     0.22 &  0.22 \\
Simulated Latencies & prom & 100.00 & 1711.93 & 1526.34 & 858.07 & 294.85 & 93.91 & 129.75 &  14.93 &   4.81 &   13.55 &   437.21 & 87.35 \\
                    & tdigest &   0.00 &    0.05 &    0.05 &   0.13 &   0.25 &  0.53 &   3.69 &  11.88 &  41.43 &  122.94 &  1062.53 &  0.00 \\
                    & hdr &   0.45 &    0.54 &    0.20 &   0.07 &   0.24 &  0.60 &   0.03 &   0.47 &   0.35 &    0.63 &     0.30 &  0.38 \\
                    & dd &   0.00 &    0.55 &    0.43 &   0.97 &   0.10 &  0.41 &   0.79 &   0.96 &   0.06 &    0.87 &     0.06 &  0.00 \\
                    & circllhist &   1.78 &    0.03 &    0.00 &   0.00 &   0.01 &  0.07 &   0.02 &   0.02 &   0.00 &    0.70 &     0.42 &  0.53 \\
\bottomrule
\end{tabular}

    \end{minipage}
    \caption{Relative errors for quantile calculation in percent.}
    \label{tab:acc}
  \end{subfigure}
\end{figure}

The first thing to note, is that the three histogram-based methods (HDR, DDSketch and circllhist)
all compute quantiles with a relative error of below 2\%.

In the body of the distribution accuracy of the circllhist is generally a little better than that of
the DDSketch and the HDR Histogram.  This is due to fact that circllhist uses fair resampling for
quantile calculations, whereas DDSketch uses paretro midpoint resampling.  It's also visible, that
DDSketch tracks min and max values separately, and reports those values exactly.

The accuracy of the t-digest is very high for the ``Uniform Distribution'' dataset and on the tails
of the ``API Latency Dataset''. However, the high quantiles of the ``Simulated Latencies'' dataset
have relative errors of more than 100\%. This example illustrates, that the t-digest does not give
any accuracy guarantees for quantile approximation after multiple merging steps have been performed.
There are a-priori error bounds for the initial data ingestion, but those are not guaranteed to hold
after merging steps. This particular dataset involves a challenging merge of 1000 batches with
highly variable distribution and size. As the authors of \cite{tdigest} write:

\begin{displayquote}
  We can force a digest formed by merging other digests to be fully merged by combining centroids
  wherever consecutive clusters taken together meet the size bound. The resulting t-digest will not
  necessarily be the same as if we had computed a t-digest from all of the original data at once
  even though it will meet the same size constraint. $[\dots]$ This loss of strictly ordering makes
  it difficult to compute rigorous error bounds.
\end{displayquote}

Our example suggests, that general rigorous error bounds are unlikely to exist.

Prometheus quantiles are only accurate for the ``Uniform Distribution'' dataset.
For the others errors of >100\% are not uncommon.

\section{Conclusion}

In this article we have introduced the circllhist as a data-structure for summarizing data which
allows accurate quantile calculations on aggregates.  To do so, we developed a general theory of
log-linear histograms, and established a-priori error bounds for reconstructed samples ($4.76\%$)
and computed quantiles ($10\%$).

We compared this data-structure to alternative data-structures which are employed in practice for
aggregated quantile calculations: Prometheus Histograms \cite{prom}, t-digest \cite{tdigest}, HDR
Histograms \cite{hdr}, and DDSketches \cite{dd}.

We have seen that, like the circllhist, also HDR Histograms and DDSketches arise as special cases of
abstract log-linear histograms as described in this document. As a result the differences between
these methods come down to configuration choices and implementation details.  In particular, all
three methods operate on essentially unbounded data ranges, offer fast insertion and merge
performance, approximate quantiles with comparable high accuracy ($<2\%$).

The Prometheus histogram, is a basic histogram data-structure that requires explicit configuration
of bin boundaries. It's reliance on numeric time-series as backing data-structures makes using large
numbers of bins impractical. With the recommended number of 10 bins, the quantile accuracy was not
competitive with the other considered methods.

The t-digest, is the only considered method that is not based on histograms. It brings competitive
insertion and merge performance, as well as very fast quantile calculations. It also delivered
precise quantile estimates for most of the considered cases. However, it does not guarantee a-priori
bounds on the relative error, as do the log-linear histogram methods. This could be seen in one of
the synthetic data-sets where large deviations of quantile values could be experienced after
multiple merges had taken place.

In comparison to the other methods, the circllhist is the oldest available method, dating back to
2011 when it was first introduced in the Circonus product. It's a one-size fit's all method that
does not require any configuration and delivers unbounded data range, full mergeability, and good
accuracy (typically $<1\%$, worst-case $<10\%$) for various summary statistics including quantile
calculations. It comes with a mature and polished implementation that delivers best-in class
performance for insertion and merge operations. The circllhist has been used in the last decade at
countless internal and third-party sites, for a highly diverse set of applications including high
volume load balancing\footnote{as part of Envoy Proxy \url{https://www.envoyproxy.io/}}, in-kernel
latency measurements\cite{HHBPF} and general application performance monitoring.

\bibliographystyle{unsrt}

\end{document}